\documentclass{article}

\usepackage{amssymb,url,amsmath}
\usepackage{bm}
\usepackage{graphicx}
\usepackage{algorithm,algpseudocode}
\usepackage{color}
\usepackage{lineno}
\usepackage{paralist}
\usepackage{hyperref}
\usepackage{caption}
\usepackage{subcaption}
\usepackage[a4paper, total={6in, 10in}]{geometry}
\graphicspath{{./}{./figures/}{./fig/}}

\newcommand{\degree}{\operatorname{degree}}

\newtheorem{theorem}{Theorem}[section]
\newtheorem{lemma}[theorem]{Lemma}

\newtheorem{problem}[theorem]{Problem}

\def\QED{\ensuremath{{\square}}}

\def\markatright#1{\leavevmode\unskip\nobreak\quad\hspace*{\fill}{#1}}
\newenvironment{proof}
  {\begin{trivlist}\item[\hskip\labelsep{\bf Proof.}]}
  {\markatright{\QED}\end{trivlist}}
  
\title{On Optimal Coverage of a Tree with Multiple Robots\thanks{
			\parbox{.1\textwidth}{
				\protect \includegraphics[trim=10cm 6cm 10cm 5cm,clip,scale=0.15]{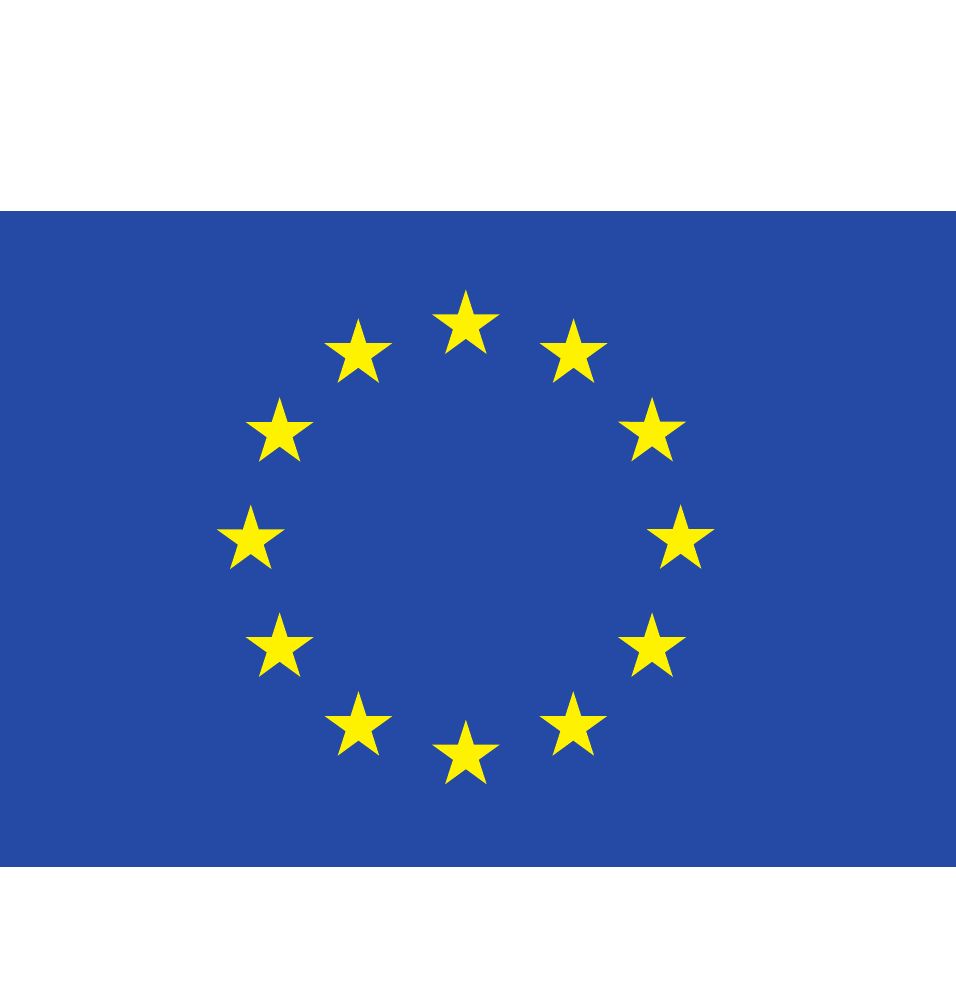}
			}
			\parbox{.9\textwidth}{ This work has also received funding from the European Union's Horizon 2020 research and innovation programme under the Marie Sk\l{}odowska-Curie grant agreement No 734922.}}
	}

\author{I. Aldana-Galv\'an\thanks{Instituto de Matem\'aticas, UNAM, Mexico City, Mexico. ialdana@ciencias.unam.mx}
\and
J.C. Catana-Salazar\thanks{Instituto de Matem\'aticas, UNAM, Mexico City, Mexico. catanas@uxmcc2.iimas.unam.mx}
\and
J.M. D\'iaz-B\'a\~nez\thanks{Departamento de Matem\'atica Aplicada II, Universidad de Sevilla, Spain. dbanez@us.es. Partially supported by project GALGO (Spanish Ministry of Economy and Competitiveness, MTM2016-76272-R AEI/FEDER,UE).}
\and
F. Duque\thanks{Instituto de Matem\'aticas, Universidad de Antioquia, Medell\'in, Colombia. rodrigo.duque@udea.edu.co}
\and
R. Fabila-Monroy\thanks{Departamento de Matem\'aticas, CINVESTAV, Mexico City, Mexico. ruyfabila@math.cinvestav.edu.mx.  Partially supported by Conacyt of Mexico grant
253261.}
\and
M.A. Heredia\thanks{Departamento de Sistemas, UAM-Azcapotzalco, Mexico City, Mexico. hvma@azc.uam.mx}
\and
A.  Ram\'irez-Vigueras\thanks{Instituto de Matem\'aticas, UNAM, Mexico City, Mexico. adriana.rv@im.unam.mx}
\and
J. Urrutia\thanks{Instituto de Matem\'aticas, UNAM, Mexico City, Mexico. urrutia@matem.unam.mx.
Partially supported by Proyect PAPIIT IN102117 from the Universidad Nacional Aut\'onoma de M\'exico. }
}

\begin{document}
\maketitle

\abstract{We study the algorithmic problem of optimally covering a tree
with $k$ mobile robots. The tree is known to all robots, and our goal is to assign a walk to each robot  in such a way that the union of these walks covers the whole tree. We assume that the edges have the same length, and that traveling along an edge takes a unit of time.
Two objective functions are considered: the cover time and the cover length. The cover time is the maximum time a robot needs to finish its assigned walk and
the cover length is the sum of the lengths of all the walks. We also consider a variant in
which the robots must rendezvous periodically at the same vertex in at most a certain number of moves. 
We show that the problem is different for the two cost functions. 
For the cover time minimization problem, we prove that the problem is NP-hard when $k$ is part of the input, regardless of whether periodic
rendezvous are required or not.
For the cover length  minimization problem, we show that it 
can be solved in polynomial time when periodic rendezvous are not required, and it is NP-hard otherwise.
}
\medskip

\noindent\textbf{Keywords}: Combinatorial Optimization; Tree Cover; Vehicle Routing; Dynamic Programming; Multi-Agent Systems.

\section{Introduction}

Operations research techniques have long been used in robotics. For instance, trajectory optimization and motion planning are application areas in which operations research is widely used \cite{hall1999,fliege2012}.
For many robotic applications, terrain coverage is a crucial task; for instance,
in search and rescue \cite{hollinger2009}, 
lawn mowing \cite{arkin2000}, 
  and surveillance by unmanned aerial vehicles \cite{acevedo2014one,diaz2017}, to name a few. 
Naturally, coverage can be sped up with multiple robots, turning the problem into a multi-robot coverage problem, in which a path has to be calculated for each robot.

Multi-agent problems have been studied for many years in combinatorial optimization. 
Computing optimal sets of routes to be covered by sets of robots to guard or cover terrains is a crucial problem in areas of research such as vehicle touting (VRP) \cite{toth2002vehicle, cordeau2007vehicle, laporte2009fifty}, location routing (LRP) \cite{nagy2007location} and multiple traveling salesman (mTSP) \cite{bektas2006multiple}), which are of prime interest in operations research.
In this paper, we will study problems in which the agents are unmanned aerial vehicles (UAVs) (also known as drones) which cover terrains, and, additionally, could meet periodically to share information.

In the terrain coverage problem, the environment can be modeled by a geometric structure, represented as the union of polygonal
obstacles, or a graph structure. The former model assumes that the robots know everything within their sight, 
and thus the problem is related to art gallery problems. This model is popular in the
computational geometry community \cite{urrutia2000}. For the latter model, the
terrain is partitioned into cells, inducing a graph whose nodes correspond to locations in the cells
and edges correspond to paths between the locations. 
In this paper, we consider the graph model, assuming that the underlying graph is a tree.  In fact, a spanning tree has 
frequently been used for multi-robot coverage \cite{hazon2008redundancy}. This appears, for instance, when considering the dual graph of a triangulation of the terrain.

Choset \cite{choset2001} provides a survey of coverage algorithms 
that distinguishes between off-line algorithms, in which a map of the work area is given to the robots, and on-line algorithms, 
in which no map is given. Two variants can also be considered according to the cost of movement: the cost can be uniform, when the move of a robot to a 
neighboring state takes unit time, or non-uniform otherwise.  
The problems introduced in this paper assume an off-line/uniform-cost scenario. We also consider the variant in which the robots are required to meet at most every $p$ 
steps for some fixed positive integer $p$.  
We assume that all robots rendezvous at the same vertex at the same time.
This rendezvous version is motivated by papers such as \cite{meghjani2011combining,hollinger2012multirobot}.

Regarding objective functions, it is frequently desirable to minimize the time at which coverage is completed. In this case, the multi-robot coverage
problem calls to compute a walk for each robot so that
the \emph{cover time}  is minimized. However, the energy efficiency of a robot's walk can also take into account the distance travelled.
In this paper we also consider another cost measure, the \emph{cover length}, meaning the sum of the lengths of all the walks needed to cover the tree. 
Note that these cost functions are different because when a robot stops and remains at a vertex, the overall cover time may increase but the cover length does not.

Before the statement of the optimization problems, we introduce some notation and assumptions for our model. 
Suppose we are given a terrain modeled by a tree $T=(V,E)$, where $V$ and $E$ are the vertices and edges sets, respectively,  that have to be covered by $k$ identical robots modeled by moving points on the tree.
The terrain is discretized by means of $n$ convex cells and each vertex of $T$ represents a cell in the terrain. 
We allow two or more robots to share a vertex or a point on an edge of $T$, without colliding or blocking each other.
We also assume that the robots walk along the edges of $T$, and that it takes one unit of time to traverse an edge.
Specifically, in one unit of time (in a step), a robot can move from one vertex to an adjacent vertex or it can stay on the same vertex.
Finally, we assume that all robots share an internal synchronized clock.
We call the journey made by a robot a \emph{walk}. Although a walk $W$ is a set of edges and vertices, we denote a walk by the sequence of $m$ visited vertices; i.e.,
 $W:=(u_1,\dots,u_m)$ where $u_i$ are vertices
of $T$ and two consecutive elements are either adjacent or equal. 

When they are equal,
the robot stays at the same vertex during the move. For convenience, we assume
that the robot stops at the last vertex of the sequence. We refer to such a sequence as a 
\emph{walk} on $T$. 
Similarly, a \emph{path} in this paper is a walk in which each
vertex is visited once. 
We denote the set of vertices of a tree $T$ (respectively a walk) as $V(T)$ (respectively $V(W)$).

The \emph{time} of a walk $W$, $t(W)$, is the number of steps that a robot needs to carry out the journey described by $W$. Formally,
the time of a walk is one less than the number of terms in $W$; $t(W)=m-1$.
The \emph{length} of a walk $W$, $l(W)$, is the number of times the robot changes its position. Formally, this is the number of times that $u_i\neq u_{i+1}$, for $i=1,\cdots,m-1$.

A \emph{strategy} is an assignment of a walk to each robot.  Formally, it is a $k$-tuple $S:=(W_1,\dots,W_k)$ of $k$ walks on $T$, where $W_i$ is the walk assigned to the $i$-th robot. We say that $S$ is a \emph{covering}
strategy if each vertex of $T$ is in at least one walk $W_i$. That is, a $\emph{covering strategy}$ is an assignment of a walk to each robot such that every vertex is visited at least once.


The robots rendezvous at time $j$ if all of them are at the same vertex at that time.
Thus, for a given strategy, a set of robots rendezvous if the $j$-th term on their walks is the same. If rendezvous is required, we also assume that the robots meet at the end of their walks. This may be necessary in scenarios in which all participants need full knowledge of a solution.

The \emph{time} of a strategy $S$, $t(S)$, is the total
time it takes for the robots to carry out the journey described
by their assigned walks.
Since the robots move in parallel we have:
\[t(S)=\max_{1\leq i \leq k}\{t(W_i)\}.\]

 The \emph{length} of a strategy $S$, $l(S)$, is the sum of the lengths
of all the walks. Thus,
\[l(S)=\sum_{i=1}^k l(W_i).\]

In this paper, we consider the following optimization problems according to the tree/off-line/uniform-cost/rendezvous framework:

\begin{itemize}
 \item \textbf{Minimum Length Covering Problem (MLCP)}
 Find a minimum length covering strategy with $k$ robots starting at given positions. 
 
 \item \textbf{Minimum Length Covering Problem with Rendezvous (MLCPR)}
   Find a covering strategy of minimum length with $k$
  robots at given starting positions, in which the robots rendezvous at most every $p$ steps, for a given positive integer $p$. 

 \item \textbf{Minimum Time Covering Problem (MTCP)}
  Find a minimum time covering strategy with $k$ robots starting at given positions.
 
 \item \textbf{Minimum Time Covering Problem with Rendezvous (MTCPR)}
  Find a covering strategy of minimum time with $k$ robots at given
starting positions, in which the robots rendezvous at most every $p$ steps, for a given positive integer $p$.
 \end{itemize}

\begin{figure}[h]
  \centering
  \includegraphics[scale=.65]{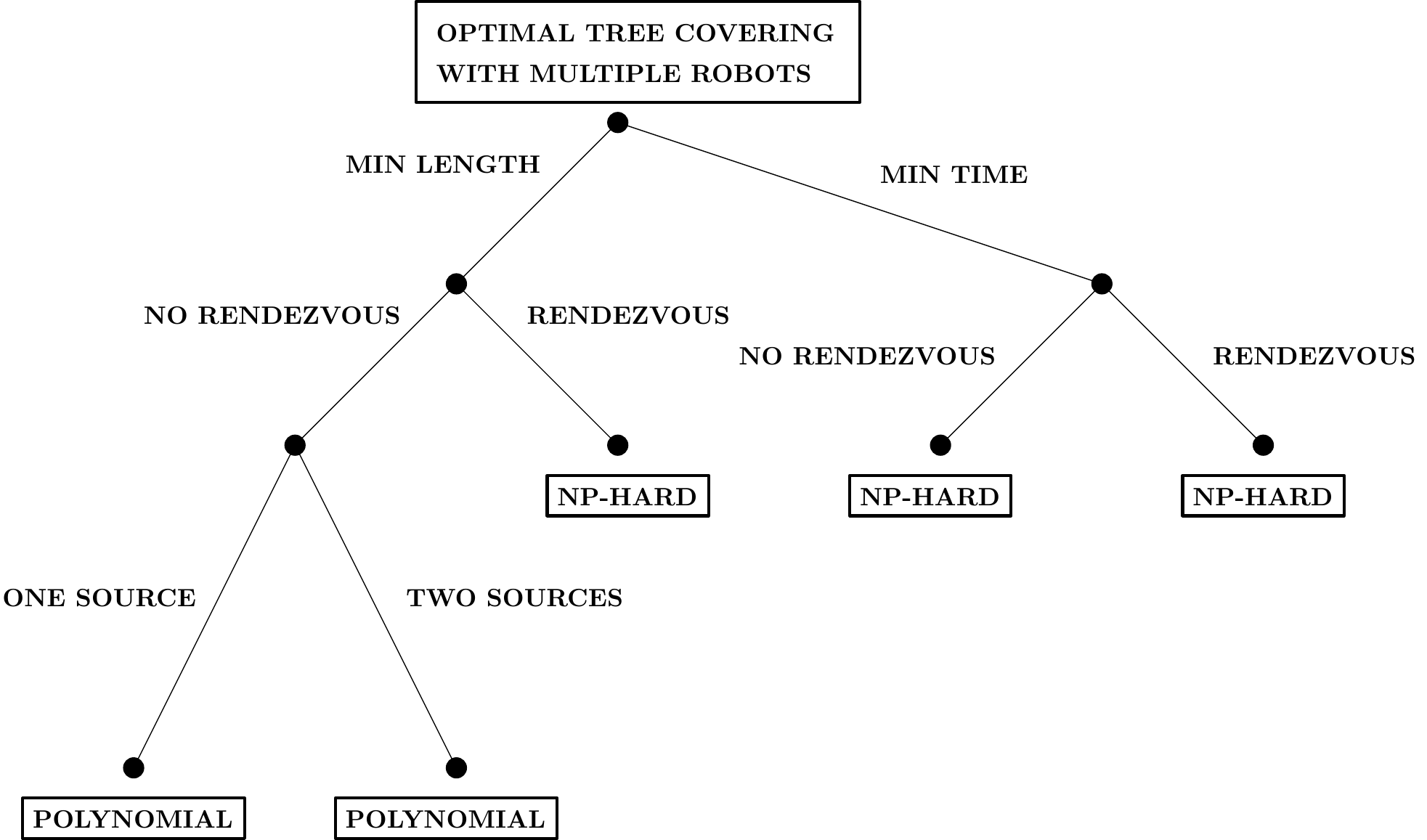}
  \caption{Complexity of the tree covering optimization problems.}
  \label{tree}
\end{figure}

It is well known that most of the optimization problems for multi-agent coverage (as well as the VRP, LRP and m-TSP) are NP-hard. In this paper, we show that the
computational complexity of the optimization problems presented above is different for the two objective functions, length and time.
We prove that MTCP is NP-hard.
Obviously, the NP-hardness of the non-rendezvous version implies NP-hardness for the rendezvous case. 
However, we show that the MLCP can be solved  in polynomial time using dynamic programming for one or two sources and then
 the complexity of the rendezvous version has to be independently established; for this case, we prove that the problem is NP-hard if $k$ is part of the input.  
Finally, the MLCP when the $k$ robots start at arbitrary $l$ positions remains open.
Our results are summarized in Figure~\ref{tree}.

The rest of the paper is organized as follows. In Section 2, we present the related literature. In Section 3, the MLCP is proved to be polynomial on $k$ (number of robots) and $n$ (number of vertices of the tree). Concretely,  we give an $O(k^2n)$ time complexity algorithm when all robots start at the same vertex and an $O(k^4n)$ time complexity algorithm when the robots start at two different locations.
In Section 4 we prove that the MLCPR is NP-hard.
In Section 5, we prove that the MTCP is NP-hard even for the non-rendezvous version. 
Finally, in Section 6 some conclusions and open problems are given.

\section{Related work}

Many methods developed in operations research have been used in robotics. An example related to the problems studied in this paper is the problem of finding an optimal routing scheme for sets of delivery vehicles departing from a central depot,  which is a problem of major interest in operations research; e.g., see \cite{clarke1964scheduling}. The delivery vehicles are required to cover a set of locations, such as stores, where some goods have to be delivered or picked up. A variant that has been studied in robotics is the so-called dynamic version of the vehicle routing problem, which deals with on-line arrival of customer demands \cite{smith2010dynamic,bullo2011dynamic,avellar2015multi}. Other well studied variations include imposing limits on the capacity of the vehicles, or time windows restricting the time during which goods have to be delivered; see \cite{solomon1987algorithms,ralphs2003parallel}.

Also, area coverage with $k$ robots is a problem that has long been investigated by the robotics research
community. 
The case $k=1$ is related to the covering salesman problem in which a robot must visit a
neighborhood of the city that minimizes the travel length \cite{arkin2000}. 
It is should be noted that a different problem is
the exploration of a completely unknown environment by a team of
mobile robots; that is, the on-line formulation -- see \cite{burgard2005} for a review. 

Regarding covering problems on a tree, a path planning algorithm for the off-line formulation of the multi-robot coverage problem is proposed in \cite{hazon2008redundancy}.  The authors split a spanning tree (previously computed to model a terrain) in such a way that each robot covers an equal portion of the tree. A closely related work on a tree, but for the on-line version, is \cite{czyzowicz2012tree}, where a swarm of mobile agents has to explore a tree, with the additional constraint that all the agents have to be close to each other; i.e. the distance between any pair of agents is at most a given positive number.

The problems studied in this paper are inspired in the use of unmanned aerial vehicles, or drones, which are being deployed on a known underlying graph to perform tasks such as the delivery of goods or surveillance. The main advantage of delivery
drones compared to regular delivery ground vehicles are that
they can operate over congested
roads without delay and without a costly human pilot. However, unmanned aerial vehicles have restrictions such as limited battery capacity or limited wireless capacity, which pose challenging problems to achieve persistence (sequences of visits to 
sites in a periodic fashion)  \cite{song2014persistent}, robustness (fault tolerance is particularly important for unmanned autonomous vehicles)  \cite{bereg2018} and, rendezvous (for example, if there is a separate team of charging robots that the drones can dock with in order to recharge)  \cite{mathew2015multirobot}. 

Note that similar problems arise in applications in which instead of drones, we use ground vehicles that move along the edges of networks representing a map of roads. 
Thus, these optimization problems with drones automatically fall into research areas such as VRP, LRP or m-TSP. Some challenging optimization problems in the interrelated areas VRP and aerial robotics have been addressed recently. For instance,  
\cite{stump2011multi} models the problem of finding sequences of visits to discrete sites in a periodic fashion such as a vehicle routing problem with time windows, and solves it using exact methods developed in the operations research community.
\cite{guerriero2014multi} presents a multi-criteria optimization model taking into account three objectives: minimizing the total distances traveled by the drones, maximizing customer satisfaction and minimizing the number of drones used. Customer satisfaction is modeled by using time window constraints.
In a recent paper \cite{agatz2018optimization}, the authors study the so-called TSP with drone, which considers the combination of a truck
and a drone in the commercial sector called ``last-mile delivery''. 
Such a combined system is found to give substantial savings compared to
the truck-only solution.

Finally, we mention a problem studied in operations research which examines the coverage of nodes of a graph using $k$ trees \cite{even2004min}. The goal is to find a set of $k$ trees such that the maximum weight of the used subtrees is minimized. The motivation arises from what the authors call the nurse station location scenario, in which nurses are assigned patients to care for. This problem is NP-hard and the authors give an algorithm to obtain a solution that is at most four times the size of the optimal solution. When the graph to be covered is a tree, a 2-approximation algorithm is given in \cite{nagamochi2004faster}; see also \cite{nagamochi2007approximating}. Note that the previous problem, although very close to ours, is quite different because we allow the robots to traverse an edge several times. To the best of the authors' knowledge, the optimization problems presented in this paper have not been studied before.


\section{Minimum Length Covering Problem (MLCP)}

Before presenting our results for the MLCP without rendezvous, we
state a useful technical result. It enables us to transform the problem 
of finding a minimum length strategy for  $T$ with $k$ robots to the problem
of finding a tuple of $k$ paths that minimizes a certain objective function. Given a graph $G$, we denote by $\mathcal{C}(G)$ the set
of its connected components and by $V(G)$ the set of its vertices.

\begin{lemma}\label{lem:struct}
Let $S:=\left (W_1,\dots,W_k \right)$ be a minimum length covering strategy for $T$ with $k$ robots.
The edge set of each $W_i$ can be decomposed into two sets: the edges of  a path $P(W_i)$ 
and the edges of a forest $F(W_i)$. Moreover, 
\begin{itemize}
 \item 

(a) each edge in $P(W_i)$ is traversed once by the $i$-th robot; 

\item 

(b) each edge in $F(W_i)$ is traversed twice by the $i$-th robot and is never traversed by any other robot;

\item 

(c) each $C \in \mathcal{C} (F(W_i))$  satisfies 
\[\left | V(C) \cap \left ( \bigcup_{i=1}^k V(P(W_i)) \right) \right |=1; \textrm{ and }\]

\item 

(d) for every pair $(j,l)$, $j \neq l$, we have that \[V(F(W_j))\cap V(F(W_l)) \subset \bigcup_{i=1}^k V(P(W_i)).\]
\end{itemize}
\end{lemma}

\begin{proof} 
We first show that a robot can traverse an edge of $T$ at most twice. 
\sloppy By contradiction, suppose that $S$ is a minimum-length covering strategy for $T$ and there exists an edge $(v,w)$ that is traversed at least three times
by the $i$-th robot. Then $W_i$ can be written as $W_i=:(Q_{1},v,w,Q_{2},w,v,Q_{3},v,w,Q_{4})$ for 
some walks $Q_{1}$, $Q_{2}$, $Q_{3}$ and $Q_{4}$. See Figure~\ref{fig:lema1}.
Note that we can replace $W_i$ by the walk $(Q_{1},v,Q_{3},v,w,Q_{2},w,Q_{4})$ and
obtain a covering strategy for $T$ of shorter length, contradicting the optimality of $S$.

\begin{figure}[ht!]
	\centering
	\begin{subfigure}[t]{0.3\textwidth}
		\centering
		\includegraphics[width=.6\linewidth]{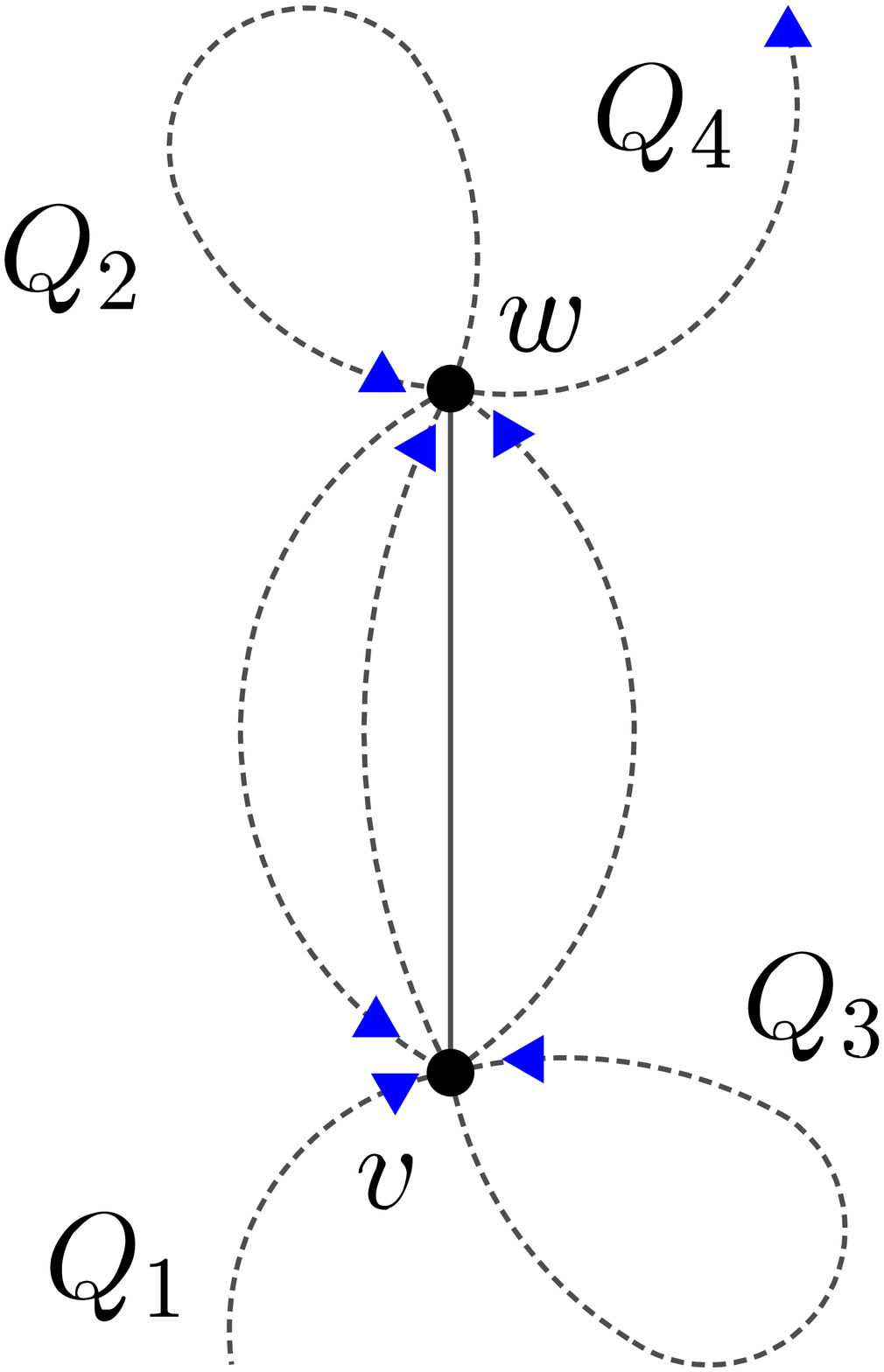}
		\captionsetup{width=.95\linewidth}
		\caption{}
		\label{fig:lema1}
	\end{subfigure}%
	\begin{subfigure}[t]{0.3\textwidth}
		\centering
		\includegraphics[width=.6\linewidth]{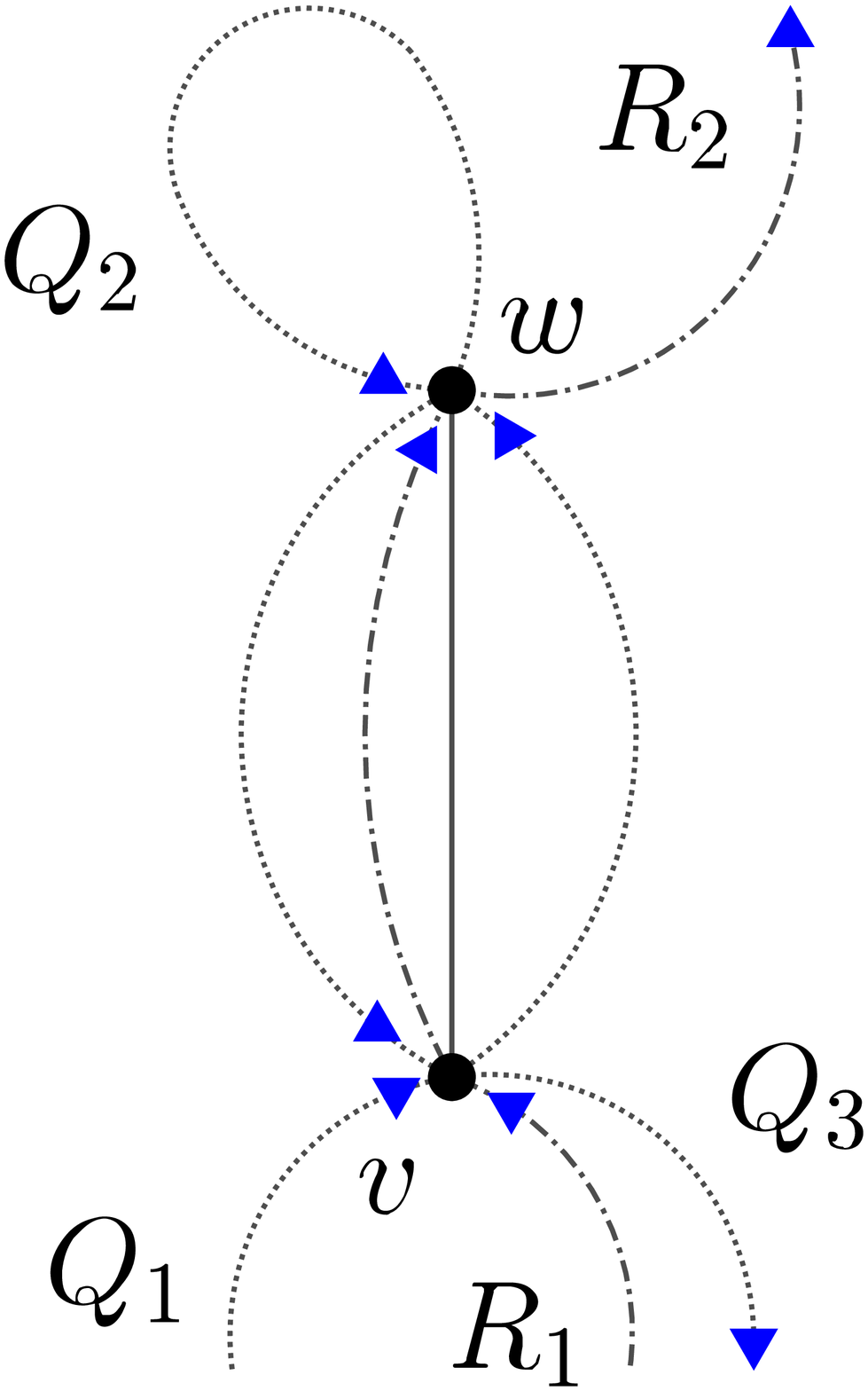}
		\captionsetup{width=.95\linewidth}
		\caption{}
		\label{fig:lema1b}
	\end{subfigure}%
	\begin{subfigure}[t]{0.3\textwidth}
		\centering
		\includegraphics[width=.6\linewidth]{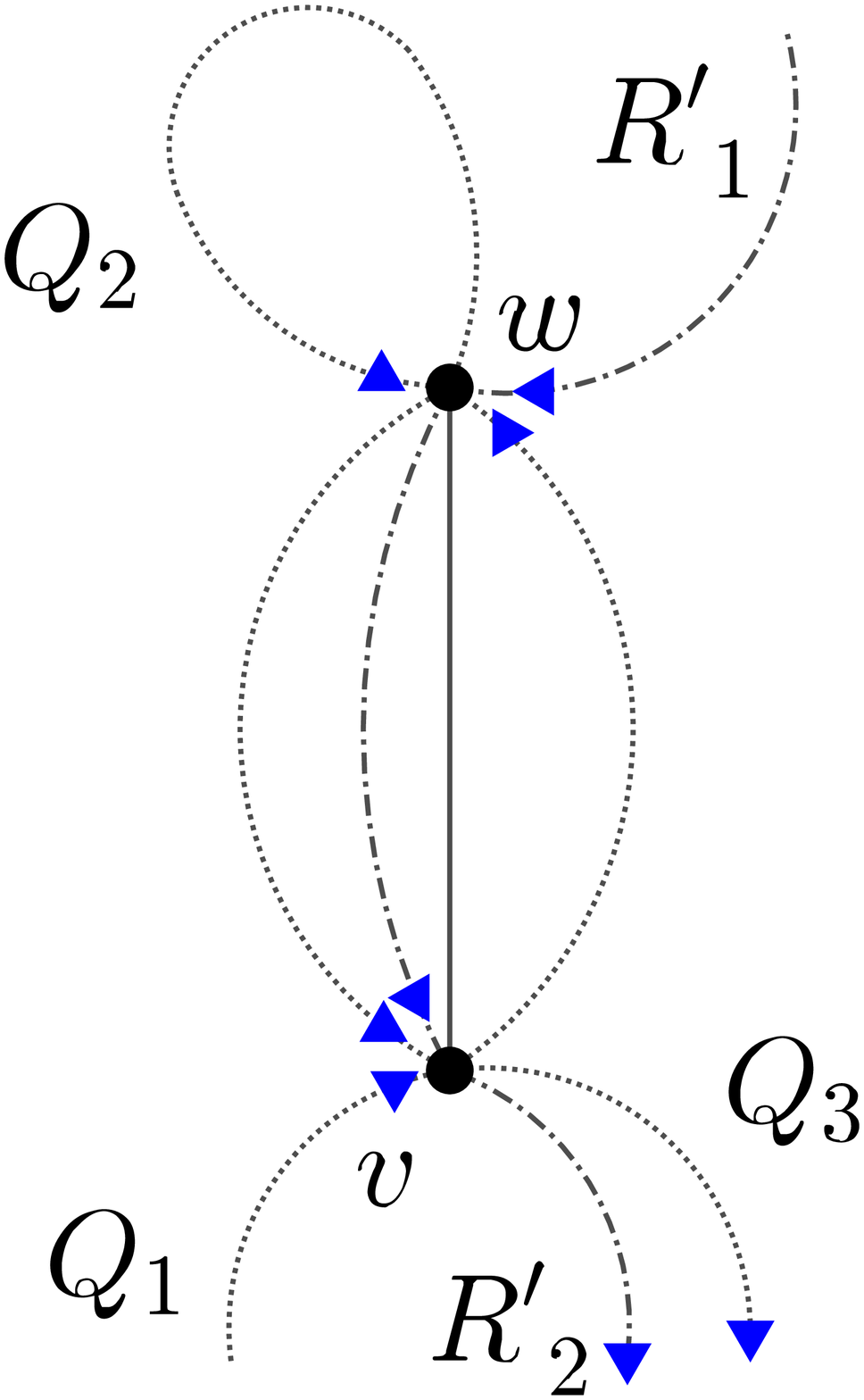}
		\captionsetup{width=.95\linewidth}
		\caption{}
		\label{fig:lema1c}
\end{subfigure}%
	\captionsetup{width=.85\linewidth}
	\caption{(a) A walk traversing the edge $(v,w)$ three times. (b),(c) show the cases when two robots traverse an edge three times.}
\end{figure}
 
 We define $P(W_i)$ and $F(W_i)$. For each $W_i$, let $P(W_i)$ be the
 subgraph of $T$ induced by the edges in $W_i$ traversed once by the $i$-th robot,
 and let $F(W_i)$ be the subgraph of $T$ induced by the edges in $W_i$ traversed twice by the $i$-th robot.
 Note that since the edges in $P(W_i)$ are traversed once by the $i$-th robot and $T$ is a 
 tree, then $P(W_i)$ is a path and $F(W_i)$ is a forest; this proves $(a)$.
 To prove $(b)$, it only remains 
 to show that no robot other than the $i$-th robot visits the edges in $F(W_i)$.
 
 Suppose that an edge $(v,w)$ in $F(W_i)$ is visited also by the $j$-th robot,
 for some $j \neq i$. Without loss of generality suppose that the $i$-th robot visits $(v,w)$ first from $v$ to $w$, 
 and afterwards from $w$ to $v$. Then $W_i$ can be written as $W_i=:(Q_{1},v,w,Q_{2},w,v,Q_{3})$ 
 for some walks $Q_1, Q_2$ and $Q_3$. We have the following cases:
 
 \begin{enumerate}
 	\item \emph{$W_j$ visits $v$ first.}
 	
 	Then $W_j$ can be written as $(R_1,v,w,R_2)$ for some walks $R_1$ and $R_2$. 
 	See Figure~\ref{fig:lema1b}. By replacing $W_i$ with $(Q_{1},v,Q_{3})$ and $W_j$ with $(R_{1},v,w,Q_2,w,R_{2})$,
 	we obtain a covering strategy for $T$ of shorter length, contradicting the optimality of $S$.
	
	\item \emph{$W_j$ visits $w$ first.}
	
	Then $W_j$ can be written as $(R'_1,w,v,R'_2)$ for some walks $R'_1$ and $R'_2$. 
	See Figure~\ref{fig:lema1c}. By replacing $W_i$ with $(Q_{1},v,Q_{3})$ and $W_j$ with $(R'_{1},w,Q_2,w,v,R'_{2})$,
	we obtain a covering strategy for $T$ of shorter length, contradicting the optimality of $S$.
 \end{enumerate}
 
 \begin{figure}[ht!]
		\centering
		\includegraphics[width=.7\linewidth]{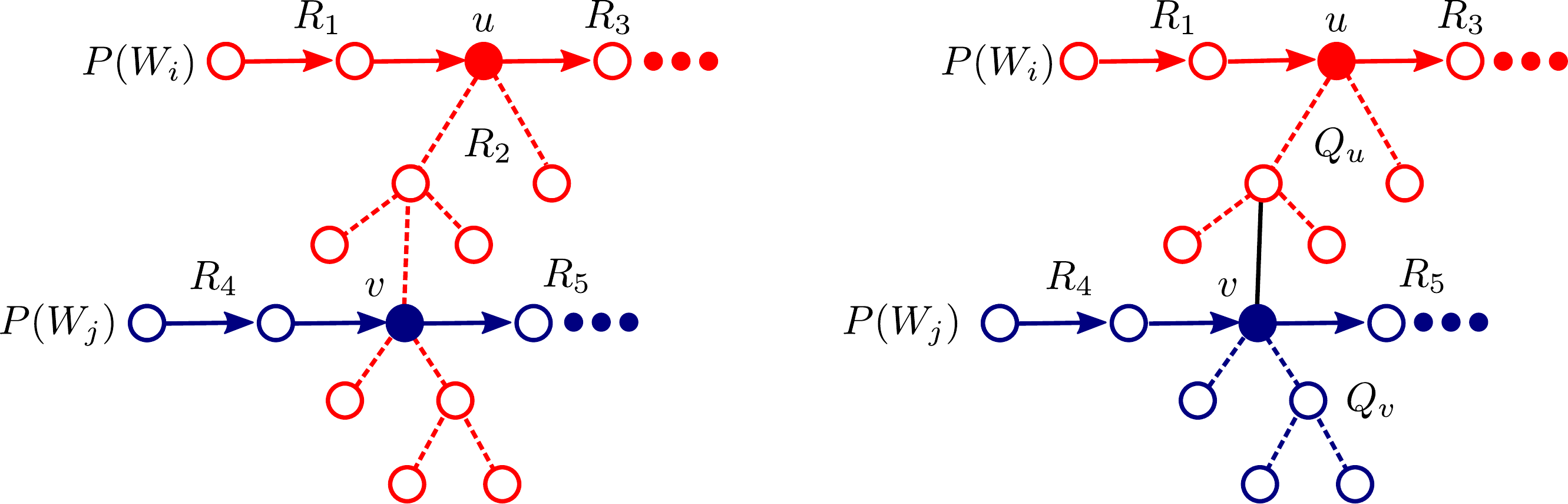}
		\captionsetup{width=.85\linewidth}
		\caption{Illustration of property (c).} 			
		\label{fig:Lemma-c}
	\end{figure}
 
We now prove $(c)$ by contradiction. Suppose that there exists a connected component
$C$ in some $F(W_i)$
such that \[X=:\left | V(C) \cap \left ( \bigcup_{j=1}^k V(P(W_i)) \right) \right |>1\] as illustrated in Figure~\ref{fig:Lemma-c}.
Since $T$ is a tree and $F(W_i)$ is edge disjoint from $P(W_i)$, we have that $\left|X \cap V(P(W_i))\right|=1$.
Let $\{u\}:=X \cap V(P(W_i))$, let $v$ be a vertex in $X$ distinct from $u$ and let
$W_j$ be the walk such that $v \in P(W_j)$. Root $C$ at $u$ and let
$T_v$ be the subtree of $C$ rooted at $v$. Let $T_u:=C\setminus T_v$. 
Let $Q_u$ and $Q_v$ be in-order traversals of $T_u$ and $T_v$, respectively.

Note that $W_i$ and $W_j$ can be written as
$W_i=:(R_1,u,R_2,u,R_3)$ and $W_j=:(R_4,v,R_5)$ for some walks $R_1,R_2,R_3,R_4$ and $R_5$.
By replacing $W_i$ with $(R_1,u,Q_u,u,R_3)$ and $W_j$ with $(R_4,v,Q_v,v,R_5)$ we obtain
a covering strategy for $T$ of length one less than the length of $S$. (In particular the edge
joining $v$ to its parent $C$ is no longer traversed.) This is a contradiction
to the optimality of $S$.

Finally, we prove $(d)$ by contradiction. Suppose that a pair $F(W_j)$ and $F(W_l)$ intersect
at a $u$ vertex not in \[\bigcup_{i=1}^k V(P(W_i)).\]
Let $C_1$ and $C_2$ be the connected components of $F(W_j)$ and $F(W_l)$ containing
$u$ respectively. See Figure~\ref{fig:Lemma-d}. Let $v$ be the vertex of $C_1$ in $P(W_j)$.
Root $C_1$ at $v$ and let $w$ be the parent of $u$ in $C_1$. Then $W_j$ and $W_l$ can be written as
$W_j=:(Q_1,w,u,Q_2,u,w,Q_3)$ and $W_l=:(R_1,u,R_2)$ for some walks $Q_1$, $Q_2$, $Q_3$, $R_1$
and $R_2$. By replacing $W_j$ with $(Q_1,w,Q_3)$ and $W_l$ with $(R_1,u,Q_2,u,R_2)$, we obtain
a covering strategy for $T$ of length one less than the length of $S$, a contradiction.

\begin{figure}[ht!]
		\centering
		\includegraphics[width=.9\linewidth]{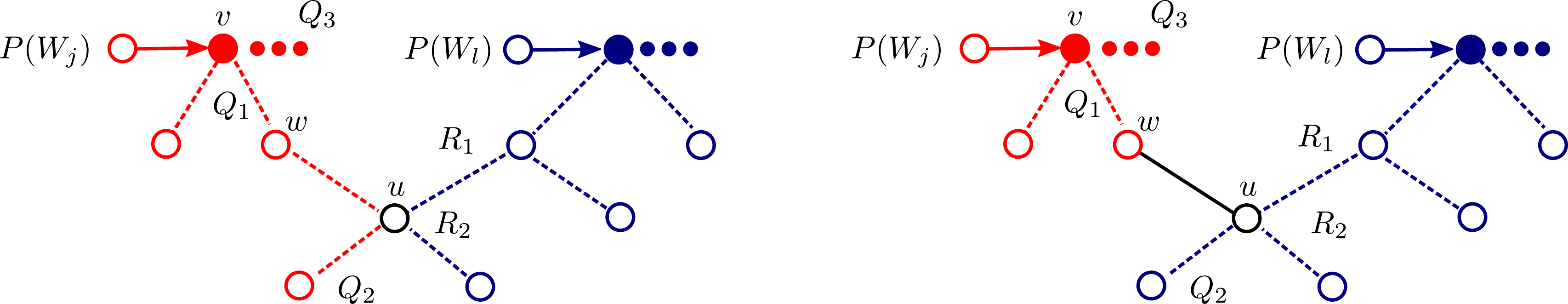}
		\captionsetup{width=.85\linewidth}
		\caption{Illustration of property (d).} 
		\label{fig:Lemma-d}
	\end{figure}

\end{proof}

Suppose now that $S:=(W_1,\dots,W_k)$ is a minimum length covering strategy for $T$.
Let $e$ be an edge of a walk $W_i$. By Lemma~\ref{lem:struct} we have the following.
If $e$ is in exactly $q$ of the paths $P(W_i$), then it is visited exactly $q$ times.
Otherwise, $e$ is visited exactly twice. 
Therefore, the cost of $S$ is given by:
\begin{eqnarray}\label{eq:cost}
 l(S) & = &\sum_{i=1}^k l(P(W_i))+2\sum_{i=1}^k |E(F(W_i)|\nonumber \\
  & = &\sum_{i=1}^k l(P(W_i))+2\sum_{i=1}^k \sum_{C \in \mathcal{C}(F(W_i))}   \left ( |C|-1 \right )\nonumber \\
 & = & \sum_{i=1}^k l(P(W_i))+2\left |V(T)\setminus \bigcup_{i=1}^k V(P(W_i)) \right|.
\end{eqnarray}
The last equality follows from $(c)$ and $(d)$ of Lemma~\ref{lem:struct}.
A remarkable property of Equation~\eqref{eq:cost} is that the cost of $S$ only depends on the paths $P(W_i)$. 
In the following lemma we show that if we are given
a tuple of $k$-paths, then we can efficiently build a covering strategy
with the same structure as in Lemma~\ref{lem:struct}.

\begin{lemma}\label{lem:struct_conv}
Let $\mathcal{P}:=(P_1,\dots,P_k)$ be a tuple of $k$ paths in $T$.
Then in $O(n)$ time we can compute a covering strategy $S:=(W_1,\dots,W_k)$
for $T$ with $k$ robots that satisfies the following. 
The edge set of each $W_i$ can be decomposed into two sets: the edges of a path $P(W_i)=P_i$ 
and the edges of a forest $F(W_i)$. Moreover, 
\begin{itemize}
 \item 

(a) each edge in $P(W_i)$ is traversed once by the $i$-th robot; 

\item 

(b) each edge in $F(W_i)$ is traversed twice by the $i$-th robot and is never traversed by any other robot;

\item 

(c) each $C \in \mathcal{C} (F(W_i))$  satisfies that 
\[\left | V(C) \cap \left ( \bigcup_{i=1}^k V(P(W_i)) \right) \right |=1; \textrm{ and }\]

\item 

(d) for every pair $j \neq l$ we have that \[V(F(W_j))\cap V(F(W_l)) \subset \bigcup_{i=1}^k V(P(W_i)).\]
\end{itemize}
\end{lemma}
\begin{proof}

Starting from $i=1$ to $k$, we construct $W_i$ as follows.
For each vertex $v$ of $P_i$, let $T_v$ be the tree rooted
at $v$ with the maximum number of edges such that 
\begin{itemize}
 \item $V(T_v) \cap V(P_i)=\{v\}$ and
 \item  $(V(T_v) \setminus \{v\}) \cap W_j= \emptyset$ for all $j <i$.
\end{itemize}
Note that $T_v$ can be computed in $O(|T_v|)$ by doing a depth-first
search starting at $v$. Let $Q_v$ be an in-order traversal
of $T_v$ starting at $v$.

Assume $P_i=:(v_1,\dots, v_m)$ and let 
\[W_i:=(v_1,Q_{v_1},\dots,v_m,Q_{v_m}).\]
Thus, $W_i$ is the walk in which
the $i$-th robot follows $P_i$, stopping at every vertex $v$ of $P_i$
to do an in-order traversal of $T_v$. 
Let $P(W_i)=P_i$ and  $F(W_i)=\bigcup_{v \in V(P_i)} T_v$.
Let $S=(W_1,\dots,W_k)$.
Since every vertex of $T$ is in some $T_v$ for some
$P_i$, $S$ is a covering strategy for $T$ with $k$ robots. By construction, $S$
satisfies $(a)$,$(b)$, $(c)$ and $(d)$. Finally, $S$ is computed in time
\[\sum_{i=1}^k \sum_{v \in P_i} O(|T_v|)=O(n).\]
\end{proof}

Lemmas~\ref{lem:struct} and~\ref{lem:struct_conv} together imply that to find a minimum covering strategy 
for $T$ with $k$ robots, it is sufficient to find a tuple of 
$k$ paths $(P_1,\dots,P_k)$ that minimizes the following expression.
\begin{equation}\label{eq:cost_paths}
 \sum_{i=1}^k l(P_i)+2\left |V(T)\setminus \bigcup_{i=1}^k V(P_i) \right|.
\end{equation}
In what follows we use this fact implicitly. We represent a covering strategy by its tuple of paths.
Each path is represented by a linked list. This
representation enable us to update our covering strategies efficiently, thus we use dynamic programming.  For brevity,
we only show how to compute the costs. Computing the corresponding strategies
can be done as in the proof of Lemma~\ref{lem:struct_conv}.

\subsection{All Robots Starting at the Same Vertex}\label{sec:one_source}
Assume that all robots start at a vertex $u$
 of $T$ and let $1 \le j \le k$ be an integer. 
Let \[\textproc{One-source}[T,u,j]\] be the cost of a minimum-length covering
strategy for $T$ with $j$ robots, all starting at $u$. 

\begin{theorem} \label{thm:k-robots-same-node}
The set  \[\left \{ \textproc{One-source}[T,u,j]: 1 \le j \le k \right \}\]
can be computed in $O(k^2 n )$ time. 
\end{theorem}

\begin{proof}
	Assume that $T$ is rooted at $u$ and the children of every vertex of $T$ are listed in some arbitrary order.
For every vertex $v \in T$, let $T_v$ be the subtree of $T$ rooted at $v$ and let $T_v[i]$ be the subgraph of $T$ 
consisting of the union of the subtrees rooted at the 
first $i$ children of $v$ together with the edges joining these children to $v$. See Figure~\ref{Treev}.
	
	\begin{figure}[ht!]
		\centering
		\includegraphics[width=.38\linewidth]{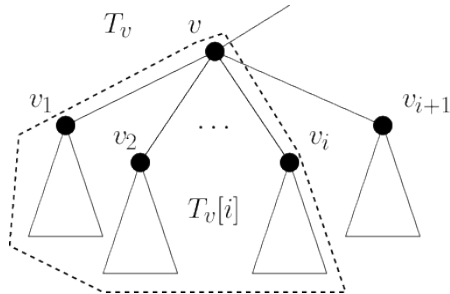}
		\captionsetup{width=.85\linewidth}
		\caption{A subtree of $T$ rooted at a vertex $v$. The subtree $T_v[i]$ is enclosed by dashed lines.} 
		\label{Treev}
	\end{figure}
	
	We use  an auxiliary table $C[v,i,j]$, where $v$ runs over all vertices
	of $T$, $1 \le i \le \degree(v)$ and $0 \le j \le k$. 
	For $j>1$, the table $C[v,i,j]$ stores the cost of a minimum exploring
	strategy for $T_v[i]$ with $j$ robots all starting at $v$. 
	For every vertex $v$ in $T$, we set $C[v,i,0]$ to be equal to twice the number of edges in $T_v[i]$. 
	Note that this is the cost of exploring $T_v[i]$  with a robot that starts and ends its walk
	at $v$. Also note that $C[u,\degree(u),j]$ is the desired $\textproc{One-source}[T,u,j]$.
	
	Let $v$ be a vertex of $T$ and let $v_1,\dots,v_m$ be its children.
	Consider an optimal covering strategy for $T_v[i]$ with $j$ robots
	all starting at $v$. Assume that
	it is represented by a tuple of paths that minimize~(\ref{eq:cost_paths}).
	By Lemma \ref{lem:struct}, if none of these paths end at a vertex of $T_{v_i}$, then $T_{v_i}$ is covered by
	one robot; this robot visits the edge $(v,v_i)$ twice. 
	If $l > 0$ of these paths end at a vertex of $T_{v_i}$,
	then each of the corresponding robots visits the edge $(v,v_i)$ once; the remaining paths
	end at a vertex of $T_v[i-1]$. Therefore,  $C[v,i,j]$ is equal to the minimum
	of 
	     \[C[v,i-1,j]+C[v_i,\degree(v_i),0]+2 \]
	and 
	     \[C[v,i-1,j-l]+C[v_i,\degree(v_i),l]+l, \] 
	for all $1 \le l \le j$.
	
 We compute $C[v,i,j]$ from bottom to top.
Having computed $C[v,i,j]$ for all vertices at height $h$ of the rooted tree, we compute these values for the vertices at height $h-1$. Since there are in total $O(k n)$ entries in  $C$ and computing each entry takes $O(k)$ time, the algorithm spends $O(k^2 n)$ time.
	
\end{proof}

\subsection{Robots starting at two vertices}

In this section, we show that the case in which the $k$ robots start at two different locations $u$ and $v$ can also be solved by dynamic programming.

Let $s,t \ge 0$ be integers. For $s,t \ge 1$ let 
\[\textproc{Two-sources}[T,u,v,s,t]\]
be the cost of an optimal covering strategy for $T$ in which $s$ robots start at $u$ and
$t$ robots start at $v$. 
Let \[\textproc{Two-sources}[T,u,v,0,t]\] be the cost of an optimal covering strategy for $T$ in which $1$ robot
starts and ends at $u$, and  $t$ robots start at $v$. 
Let \[\textproc{Two-sources}[T,u,v,s,0]\] be the cost of an optimal covering strategy for $T$ in which $1$ robot 
starts and ends at $v$, and  $s$ robots start at $u$.

To solve the problem of two starting locations,
we first need some definitions and a pair of technical results. Let $u$ and $v$ be two vertices of $T$. We rename the vertices of the path $\gamma$ from $u$ to $v$ in $T$ as $x_1, x_2,\dots, x_m$; i.e., $\gamma := (u=x_1, x_2,\dots, x_m=v)$. 
 Let $F$ be the forest obtained by removing all the edges in $\gamma$ from $T$; 
 let $T_1,\dots, T_m$ be the connected components of $F$, 
 where $T_i$ is the component containing $x_i$. Assume that $T_i$ is rooted at $x_i$.
 Let $T_{u,v}[i]$ be the tree consisting of the union of the subpath $\gamma_i:=(x_1, \ldots, x_i)$ and the trees $T_1, \dots, T_i$ (see Figure~\ref{Path_P}).

 \begin{figure}[h]
		\centering
		\includegraphics[width=.46\linewidth]{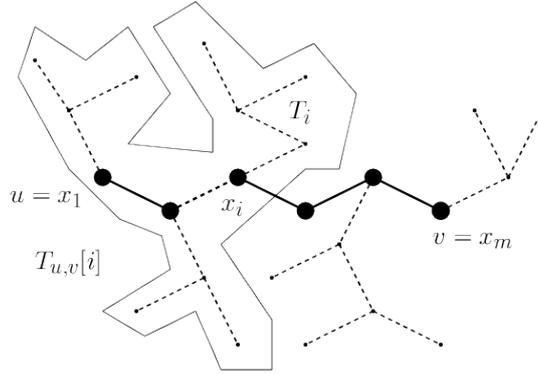}
		\captionsetup{width=.85\linewidth}
		\caption{The tree $T_{u,v}[i]$ is in the polygonal region. 
		The path joining $u=x_1$ to $v=x_m$ is marked with solid black line segments. $T_1,\dots,T_m$ 
		are marked with dashed black line segments. }
		\label{Path_P}
	\end{figure}

For integers $1 \le k' \le k$, $0 \le j \le k'$ and $1 \le i \le m$,  let \[\textproc{Destination-path}[T,u,v,x_i,k',j]\]
be cost of an optimal covering strategy for $T_{u,v}[i]$ with $k'$ robots
 all starting at $u$ in which at least $j$ robots end their walks at $x_i$, and
 let \[\textproc{Destination-path}[T,u,v,x_i,0,0]\]
be the cost of an optimal covering strategy for $T_{u,v}[i]$ with $1$ robot that
starts and ends at $u$.

\begin{theorem} \label{thm:at-least}
The set \[\left \{ \textproc{Destination-path}[T,u,v,x_i,k',j]\colon 0 \le j\le  k' \le k, \textrm{ and }  1 \le i \le m \right \}\] 
can be computed in  $O(k^3n)$ time.
 \end{theorem}
\begin{proof}

First we use Theorem~\ref{thm:k-robots-same-node} to compute  the set
\[\left \{ \textproc{One-source}[T_i,x_i,l]\colon 1 \le i \le m \textrm{ and } 1 \le l \le k \right \}.\]
This can be done in  $\sum_{i=1}^m O(k^2 |T_i|)=O(k^2 n)$  time. 
We now show how to compute $\textproc{Destination-path}[T,u,v,x_i,k',j]$
depending on the values of $i$ and $j$. Consider an optimal covering strategy
for  $T_{u,v}[i]$ with $k'$ robots
 all starting at $x_1=u$ in which at least $j$ robots end their walks at $x_i$ and assume
 that it is represented by a tuple of $k'$ paths that minimize \eqref{eq:cost_paths}. We have the following cases:
\begin{itemize}
 \item $i=1$.
 
	In this case, at least $l \ge j$ paths consist only of the vertex $x_1$ and  the remaining 
	paths end at a vertex of $T_1$. Therefore, 
	$\textproc{Destination-path}[T,u,v,x_1,k',j]$ is equal to the minimum of  
	\[\textproc{One-source}[T_1,x_1,k'-l],\]
	overall $j \le l \le k'.$

 \item $i > 1$ and $j=0$.
 
        In this case, no path is required to end at $x_i$. If no path ends at a vertex of $T_i$, 
        then $T_i$ is explored by a single robot that visits the edge $(x_{i-1},x_i)$ twice. 
        If $l > 0$ paths end at a  vertex of  $T_i$, then each of the corresponding
        robots visits the edge $(x_{i-1},x_i)$ once. Therefore, 
        $\textproc{Destination-path}[T,u,v,x_i,k',0]$ is equal to the minimum of 
       
         \[\textproc{Destination-path}[T,u,v,x_{i-1},k',0]+\textproc{One-source}[T_i,x_i,0]+2\]
         and 
        \[\textproc{Destination-path}[T,u,v,x_{i-1},k',l]+\textproc{One-source}[T_i,x_i,l]+l,\]
       over all $1 \le l \le k'$.

\item $i > 1$ and $j > 0$.

        Suppose that  $l \ge j$ of the paths end at a vertex of $T_i$. Thus, 
        each of their corresponding robots visits the edge 
        $(x_{i-1},x_i)$ once, $j$ of them  end their walks at $x_i$, 
        and $l-j$ of them may end their walks at a vertex of $T_i$ different from 
        $x_i$. Therefore, 
        $\textproc{Destination-path}[T,u,v,x_i,k',j]$ is equal to the minimum of 
       \[\textproc{Destination-path}[T,u,v,x_{i-1},k',l]+\textproc{One-source}[T_i,x_i,l-j]+l,\]
       over all $j \le l \le k'$.
\end{itemize}
Using dynamic programming, each  entry  can be computed  in $O(k)$ time.
Thus, computing the whole table can be done in $O(k^3n)$ time.
\end{proof}

\begin{lemma}\label{lem:no_crossing}
Let $(x,y)$ be an edge of an optimal covering strategy for $T$ that is visited by at least two different robots. Then all robots traverse $(x,y)$ in the same direction;
either from $x$ to $y$ or from $y$ to $x$.
\end{lemma}
\begin{proof}
By contradiction, suppose that a robot traverses $(x,y)$ from $x$ to $y$ and
that another robot traverses $(x,y)$ from $y$ to $x$.
Let $(W_1,x,y,W_2)$
be the walk traversed by the first robot, and $(W_1',y,x,W_2')$ be the walk traversed by the second robot. 
If we replace the first walk by $(W_1,x,W_2')$ and 
the second walk by $(W_1',y,W_2)$, 
we obtain a covering strategy of smaller length and the result follows.
\end{proof}

The next theorem shows how to find an optimal covering strategy for $T$ for the case when there
are two starting locations, $u$ and $v$. 

\begin{theorem} \label{thm:2-vertices}
The set \[\left \{\textproc{Two-sources}[T,u,v,s,t]: 0 \le s,t \le k\right \}\]
can be computed in $O(k^{4} n)$  time.
\end{theorem}

\begin{proof}
First, we use Theorem~\ref{thm:k-robots-same-node} to compute in $O(k^2 n)$ time the set
\[\left \{ \textproc{One-source}[T_i,x_i,l]\colon 1 \le i \le m \textrm{ and } 1 \le l \le k \right \}.\]
We now use Theorem~\ref{thm:at-least} to compute in $O(k^3n)$ time the sets
\[\left \{ \textproc{Destination-path}[T,u,v,x_{i},k',j]\colon 0 \le j\le  k' \le k, \textrm{ and }  1 \le i \le m \right \}\] 
and
\[\left \{ \textproc{Destination-path}[T,v,u,x_{i},k',j]colon 0 \le j\le  k' \le k, \textrm{ and }  1 \le i \le m \right \}.\] 
Note that in the first set the robots start at $u$, and in the second set the
robots start at $v$.

For every triple of integers $1 < i \le m$, $0 \le s,t \le k$, let $C[i,s,t]$ be the cost of an optimal
covering strategy for $T$ in which $s$ robots start at $u$, $t$ robots
start at $v$ and $(x_{i-1},x_{i})$ is the last edge of $\gamma$ that is visited
by a robot starting at $u$.  Let $C[1,s,t]$ be the cost when no robot starting at $u$ visits an edge of $\gamma$.
Note that $\textproc{Two-sources}[T,u,v,s,t]$ is equal to the minimum of $C[i,s,t]$ over all $1 \le i  \le m$.

By Lemma~\ref{lem:no_crossing}, we know that in an optimal covering strategy, at most one vertex of $\gamma$ is visited by both a robot
starting at $u$ and a robot starting at $v$. Therefore, 
$C[i,s,t]$ can be computed from $\textproc{Destination-path}[T,u,v,x_i,s,j]$ and
$\textproc{Destination-path}[T,v,u,x_i,t,j]$ as follows:
\begin{itemize}
 \item \emph{$i=1$.}
 
 In this case, no robot starting at $u$ traverses the edge $(x_1,x_2)$. Therefore, at
 least one robot starting at $v$ visits $x_2$, and $T_1$ is visited by the $s$
 robots starting at $u$ and possibly some robots starting at $v$. Thus,
 $C[1,s,t]$ is equal to the minimum of 
 \[\textproc{One-source}[T_1,x_1,s+j]+\textproc{Destination-path}[T,v,u,x_2,t,j]+j,\]
 over all  $0 \le j \le t$. 
 \item \emph{$i=m$.}
 
 In this case, no robot starting at $v$ traverses the edge $(x_{m-1},x_m)$. Therefore, at
 least one robot starting at $u$ visits $x_m$ and $T_m$ is visited by the $t$
 robots starting at $v$ and possibly some robots starting at $u$. 
 Thus, $C[m,s,t]$ is equal to the minimum of 
 \[\textproc{One-source}[T_m,x_m,t+j]+\textproc{Destination-path}[T,u,v,x_{m-1},s,j]+j,\]
 over all  $0 \le j \le t$. 
 
 \item \emph{$1 < i < m$.}
 
  In this case, at least one robot starting at $u$ enters $T_{i}$. If a robot
  starting at $u$ visits the edge $(x_{i-1},x_{i})$ twice, then, by Lemma~\ref{lem:no_crossing}, no robot
  starting at $v$ can enter $T_{i}$.
  Therefore, $C[i,s,t]$ is equal to the minimum of 
  \[\textproc{Destination-path}[T,u,v,x_{i},s,0]+\textproc{Destination-path}[T,v,u,x_{i+1},t,0]\]
  and
  \begin{align}
   & \textproc{Destination-path}[T,u,v,x_{i-1},s,j]+\textproc{One-source}[T_i,x_i,j+l]+ \nonumber \\
   & \textproc{Destination-path}[T,v,u,x_{i+1},t,l]+j+l\nonumber,   \end{align}
   over all $1 \le j \le s$ and $0 \le l \le t$. 
\end{itemize}

Using dynamic programming, each  entry  can be computed  in $O(k^2)$ time (we have two indices $j$ and $l$).
Thus, computing the whole table can be done in $O(k^4n)$ time.

\end{proof}

\section{Minimum Length Covering Problem with Rendezvous (MLCPR)}

In this section, we consider that case in which the robots have to rendezvous at most every $p>0$
steps. We assume that all robots start at a given vertex and end at a common vertex. We prove that 
the problem of finding a minimum-length covering strategy with rendezvous
on a tree $T$ is NP-hard. Let us formalize the corresponding decision problem (Length Covering Strategy with Rendezvous).

\begin{problem}[LCSR]
 Let $T$ be a tree and let $u$ be a vertex of $T$. Let $k, \ell$ and $p$ be positive
 integers. Decide whether there exists a covering strategy $S$ for $T$, with $k$ robots starting at $u$, such 
 that the robots rendezvous at most every $p$ steps and $l(S)$ is at most $\ell$.
\end{problem}

We can prove that the LCSR problem is NP-complete by a reduction from 3-PARTITION. Thus, MLCPR is NP-hard.

\begin{problem}[3-PARTITION]
Let $B$ be a positive integer. Let $A$ be a set of $3m$ positive integers such that
$\frac{B}{4} < a < \frac{B}{2}$ for all $a \in A$, and $\sum_{a \in A} a=mB$.
The 3-PARTITION problem asks whether $A$ can be partitioned into $m$ sets $A_1,\dots, A_m$
such that for each $1 \le i \le m$, $\sum_{a \in A_i} a=B$.
\end{problem}
Note that in such a partition, each $A_i$ must consist of three elements. This problem
is known to be strongly NP-complete~\cite{multiprocessor}. This implies
that it remains NP-complete even when $B$ is bounded by a polynomial on
$m$. The following result can be stated.

\begin{theorem}
 The LCSR Problem is NP-complete.
\end{theorem}

\begin{proof}
 Given a strategy, it can be verified in polynomial time whether it satisfies that the robots rendezvous at most every $p$ steps. Since the length of a given strategy can also be computed in polynomial time, we have that LCSR is in NP. 
 
 Let $(A,B)$ be an instance of $3$-PARTITION such that
 $|A|=3m$ and $B$ is bounded by a polynomial on $m$.
 In the following, we construct (in polynomial time) an
 instance of LCSR that has a solution if and only if $A$ admits
 a $3$-PARTITION.
 
Given $A:=\{a_1,\dots, a_{3m}\}$, we consider the following tree $T$. See Figure \ref{np-length}. Let $P_1, \dots,P_{3m}$ be $3m$ paths, starting at the same vertex $u$, such that $l(P_i)=a_i$.
 Let \[T'= \bigcup_{i=1}^{3m} P_i.\]
 Let $Q=(v_1,\dots,v_{3B+4})$ be a path of length $3B+3$. 
 Finally, let \[T:=T'\cup Q \cup (u,v_1).\]
  Note that $T'$ has $mB$ edges and since $B$ is bounded by a polynomial on $m$, $T$ can be constructed
 in polynomial time. Let $k:=m+1$, $p:=2B+2$ and 
 \[\ell=m(2B+2)+(m+1)(2B+2)+(2B+2)=(2m+2)(2B+2)\]
 
 \begin{figure}[ht!]
		\centering
		\includegraphics[width=.85\linewidth]{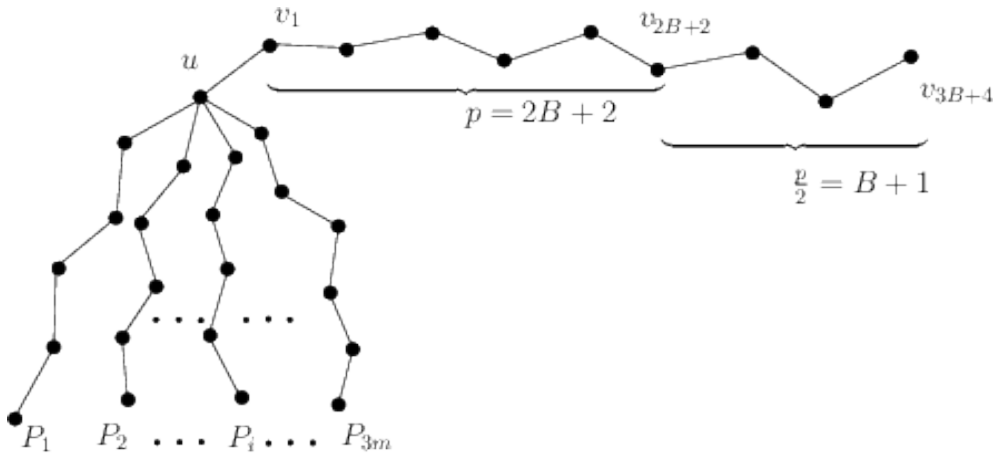}
		\captionsetup{width=.85\linewidth}
		\caption{The construction of the instance T. }
		\label{np-length}
	\end{figure}

 We claim that $T$ has a covering strategy with $k$ robots starting at $v_1$, that rendezvous at most
 every $p$ steps, and of length at most $\ell$ if and only if $A$ admits a $3$-PARTITION.
 Recall that all robots must rendezvous at the end of their corresponding walks.

 First, let $A_1, \dots, A_m$ be a $3$-partition of $A$. 
 The covering strategy of $T$ is as follows. Consider $m$ robots covering $T'$ with rendezvous at $v_1$ with length $2mB+2m$ (using the $3$-partition. The $i$-th robot covers the 3 paths associated with $A_i$) while a robot waits at $v_1$ for $2mB+2m$ steps. After that, the $m+1$ robots meet at $v_{2B+2}$ and one robot visits the last vertex of $Q$ and joins the rendezvous at $v_{2B+2}$ again.
 It is easy to see that the length of such strategy is exactly $\ell$. 

Conversely, 
let $S$ be a covering strategy of $T$ with $k$ robots starting at $v_1$  that rendezvous at most every $p$ steps, and of length at most $\ell$. We prove that $A$ admits a $3$-PARTITION.
 
 We claim that in $S$, all the robots reach the vertex $v_{2B+2}$. Suppose, on the contrary, that at least one robot
 does not visit $v_{2B+2}$. Thus, since the path from $v_{2B+2}$ to $v_{3B+4}$ is of length $B+1$, the robot visiting the last vertex
 of $Q$ is not able to rendezvous with the other robots again, a contradiction. 
Thus, $S$ spends at least length 
 $(m+1)(2B+2)+(2B+2)$ to cover $Q$.
 
%

On the other hand, we also claim that the edges of $T'$ are visited at least twice. Suppose, on the contrary, that an edge of $P_i$ is visited once. Therefore, the robots reach their final position at $P_i$. Since the length of $P_i$ is less than $\frac{B}{2}$, then at least $mB-\frac{B}{2}$ edges of $T'\setminus P_i$ must be visited twice.  Moreover, since 
 the robots do not reach their final position at some vertex of $Q$, then the edges of $Q$
 from $v_1$ to $v_{2B+2}$ must be visited at least twice by all the robots and the edges from $v_{2B+2}$ to $v_{3B+3}$ must be visited twice by at least one robot. Note that in this case, all the robots visit $(u,v_1)$ at least once. Thus, the length of $S$ is at least 
 \[2(m+1)(2B+2)+(2B+2)+(m+1)+2(mB-\frac{B}{2})>\ell,\] a contradiction.
 Therefore, the claims above imply that \[l(S)\ge 4mB+4B+2m+4\] and then at least one robot does not visit $(u,v_1)$ (due to the length constraint).

As a consequence, when a robot enters $T'$ through $(v_1,u)$, in order
 to rendezvous with the other robots it must exit $T'$ through $(v_1,u)$ again.
 Moreover, it must rendezvous with the other robots in at most $2B+2$ steps; thus the robot
 visits at most $3$ endpoints of the paths $P_1,\dots,P_m$, and afterwards exits through $(u,v_1)$.
 Since there are $3m$ such paths, this happens exactly $m$ times. Partition $A$ according to these visits
so that a set $\{a_i,a_j,a_k\}$ in this partition corresponds to a robot
 entering $T'$ and visiting the endpoints of $P_i$, $P_j$ and $P_k$ before exiting $T'$ to
 rendezvous with the other robots. Note that  for all $A_i$, \[\sum_{a \in A_i}a\le B.\]
 Since  \[\sum_{a \in A}a =mB,\] this implies that for all $A_i$, \[\sum_{a \in A_i}a = B,\] 
 and $A_1, \dots, A_m$ is a $3$-partition of $A$.

 \end{proof}
%
%
%
%
%
%
%
%
%
%
%
%
%
%
%

\section{Minimum Time Covering Problem (MTCP)}
We prove that the problem of computing a minimum time covering strategy
is NP-hard, regardless of whether periodic rendezvous are required.
The corresponding decision problem (time covering strategy) is as follows.

\begin{problem}[TCS]
 Let $T$ be a tree, with $k$ robots at given starting positions and let $t$ be a positive
 integer. Decide whether there exists a covering strategy $S$ for $T$ with these robots so that $t(S)$ is at most $t$.
\end{problem}


 

To prove the NP-completeness of the TCS problem, we also use a reduction from 3-PARTITION. 

\begin{theorem}
 The TCS Problem is NP-complete.
\end{theorem}
\begin{proof}
 Computing the time of a given covering strategy for $T$ can be done
 in polynomial time; thus we have that TCS is in NP. 
 
 Let $(A,B)$ be an instance of $3$-PARTITION such that
 $|A|=3m$ and $B$ is bounded by a polynomial on $m$.
 We construct (in polynomial time) an
 instance of TCS that has a solution if and only if $A$ admits
 a $3$-PARTITION.
 
 The instance of  TCS is as follows. Given $A:=\{a_1,\dots, a_{3m}\}$, let $P_1, \dots,P_{3m}$ be $3m$ paths,
 all starting at the same vertex $u$, such that $l(P_i)=a_i$.
 Let $L=2\sum_{i=1}^{3m} a_i$. Let $Q_1,\dots,Q_m$ be $m$ paths,
 each of length $L$ and starting at $u$. 
 Finally, let \[T:=\left ( \bigcup_{i=1}^{3m} P_i \right ) \cup \left( \bigcup_{i=1}^m Q_i \right ),\]
 $k:=m$ and $t=L+2B$. Note that since $B$ is bounded by a polynomial on $m$, $T$ can be constructed
 in polynomial time. We claim that $T$ has a covering strategy of time at most $t$,
 with $k$ robots all starting at $u$ if and only if $A$ admits a $3$-partition.
  
   Let $A_1,\dots, A_m$ be a $3$-partition of $A$. Denote by $P^{-1}$ the path $P$ traveled in opposite direction. 
   For $1 \le i \le m$, let $A_i:=\{a_{i_1},a_{i_2}, a_{i_3}\}$, let $x_i$ be the last
 vertex of $P_i$ and let $P_i'$ be the subpath of $P_i$ formed by its internal vertices.
 Thus $P_i=uP_i'x_i$.  Now, for $1 \le i \le m$, let \[W_i:=P_{i_1}P_{i_1}'^{-1}P_{i_2}P_{i_2}'^{-1}P_{i_3}P_{i_3}'^{-1}Q_i.\]

 
 Notice that each  $W_i$ spends time $L+2B$. Therefore, $(W_1,\dots,W_m)$ is a covering strategy
 for $T$ of time at most $t$, with $k$ robots all starting at $u$.
 
Conversely, let $(W_1,\dots,W_m)$ be a covering strategy for $T$ of time at most $t$, with $k$ robots
 all starting at $u$. We claim that a robot visits at most one leaf of the path $Q_i$. Otherwise,
 its walk would spend at least $3L > L+2B=t$ time. Note that $L=2Bm>B$. The above property allows us to match the $m$ robots with the paths $Q_i$.
Now, let $q(i)$ be the index
 such that the $i$-th robot visits the last vertex of $Q_{q(i)}$. 
 It is easy to see that the $i$-th robot cannot end at a vertex not in $Q_{q(i)}$, otherwise
 the time of its walk $W_i$ would be at least $2L+1=L+2Bm+1>t$ for $m>0$. We may assume that the walk $W_i$
 ends at the end-vertex of $Q_{q(i)}$, because before reaching
 this vertex, it visits all the interior vertices of $Q_{q(i)}$.
 We may also assume that no two robots visit the same path $P_i'x_i.$
 Therefore, each walk $W_i$ is of the form \[W_i:=P_{i_1}P_{i_1}'^{-1}P_{i_2}P_{i_2}'^{-1}P_{i_3}P_{i_3}'^{-1}\cdots P_{i_s}P_{i_s}'^{-1}Q_{q(i)},\]
 for some indices $i_1, i_2,\dots, i_s$. Since 
 \[l(P_{i_1}P_{i_1}'^{-1}P_{i_2}P_{i_2}'^{-1}\cdots P_{i_s}P_{i_s}'^{-1}u)\le 2B\] and
 \[l(P_{j}P_{j}'^{-1}u) > \frac{B}{2} \quad (a_i>\frac{B}{4}),\]
 we have that $s=3$.
 For $1\le i \le m$, let $A_i:=\{a_{i_1},a_{i_2}, a_{i_3}\}$.
 Now, \[a_{i_1}+a_{i_2}+a_{i_3}=l(P_{i_1}P_{i_1}'^{-1}P_{i_2}P_{i_2}'^{-1}P_{i_3}P_{i_3}'^{-1}u)/2\le B.\]
 Finally, since \[\sum_{i=1}^m a_i=mB,\] actually
 \[a_{i_1}+a_{i_2}+a_{i_3}=B,\] and the result follows.
 
\end{proof}
%
%
%
%
 \section{Conclusions}

In this paper, we have studied some optimization problems for covering a terrain (modeled by a tree) by using a team of $k>1$ robots. 
We addressed two variants of the problem, minimizing the time it takes to cover the terrain, and minimizing the total distance traversed by the robots.
We also considered the problem when a periodic rendezvous is or is not required. We showed that for the non-rendezvous case, the two variants are different.  The cover time problem is NP-hard (this implies the NP-hardness for the rendezvous case) while the cover length problem is polynomial. 
Moreover, we proved that the cover length variant with rendezvous
becomes NP-hard. 
 
We showed how to efficiently solve the Minimum Length Covering Problem (MLCP) when the robots start at one or two given vertices. The main open question is therefore whether the MLCP starting at $s$ arbitrary positions is NP-hard when $s$ is part of the input.
Finally, from a practical standpoint, 
efficient heuristics with good performance should be of interest for the hard covering problems illustrated in Figure~\ref{tree}.
%



\bibliographystyle{plain}
\bibliography{exploringbib}
\end{document}